\documentclass[conference,letterpaper]{IEEEtran}

\usepackage[utf8]{inputenc} 
\usepackage[T1]{fontenc}
\usepackage{url}
\usepackage{ifthen}
\usepackage{cite}
\usepackage{enumerate}
\usepackage[cmex10]{amsmath} 
\usepackage{cite}
\usepackage{algorithmic}
\usepackage{graphicx}
\usepackage{textcomp}
\usepackage{dsfont}
\usepackage{mathtools}

\DeclarePairedDelimiter\floor{\lfloor}{\rfloor}
\usepackage{lipsum}
\usepackage{amssymb, amsmath, amsthm, amsfonts}
\usepackage{tikz}
\usepackage{xcolor}
\usetikzlibrary{trees}
\newtheorem{thm}{Theorem}[section]
\newtheorem{lem}[thm]{Lemma}
\newtheorem{prop}[thm]{Proposition}

\ifCLASSOPTIONcompsoc
\usepackage[caption=false, font=normalsize, labelfont=sf, textfont=sf]{subfig}
\else
\usepackage[caption=false, font=footnotesize]{subfig}
\fi

\def\BibTeX{{\rm B\kern-.05em{\sc i\kern-.025em b}\kern-.08em
		T\kern-.1667em\lower.7ex\hbox{E}\kern-.125emX}}

\def\0{\mathbf{0}}

\def\E{\mathbb{E}}
\def\P{\mathbb{P}}
\def\R{\mathbb{R}}
\def\Z{\mathbb{Z}}

\def\cQ{\mathcal{Q}}
\def\cR{\mathcal{R}}

\def\G{\Gamma}

\begin{document}
	\title{Probabilistic Forwarding of Coded Packets on Networks} 
	
	\author{%
		\IEEEauthorblockN{B.~R.~Vinay Kumar and Navin Kashyap}
		\IEEEauthorblockA{Department of Electrical Communication Engineering\\
			Indian Institute of Science, Bengaluru, India\\
			Email: \{vinaykb, nkashyap\}@iisc.ac.in}
	}
		
	\maketitle

	\begin{abstract}
		We consider a scenario of broadcasting information over a network of nodes connected by noiseless communication links. A source node in the network has $k$ data packets to broadcast, and it suffices that a large fraction of the network nodes receives the broadcast. The source encodes the $k$ data packets into $n \ge k$ coded packets using a maximum distance separable (MDS) code, and transmits them to its one-hop neighbours. Every other node in the network follows a probabilistic forwarding protocol, in which it forwards a previously unreceived packet to all its neighbours with a certain probability $p$. A ``near-broadcast'' is when the expected fraction of nodes that receive at least $k$ of the $n$ coded packets is close to $1$. The forwarding probability $p$ is chosen so as to minimize the expected total number of transmissions needed for a near-broadcast. In this paper, we analyze the probabilistic forwarding of coded packets on two specific network topologies: binary trees and square grids. For trees, our analysis shows that for fixed $k$, the expected total number of transmissions increases with $n$. On the other hand, on grids, we use ideas from percolation theory to show that a judicious choice of $n$ will significantly reduce the expected total number of transmissions needed for a near-broadcast. 
	\end{abstract}
	
	
	\section{Introduction}\label{intro}
	The Internet of Things (IoT) involves different types of physical devices --- sensors, actuators, routers, mobiles etc. ---  communicating with each other over a network. Each node in the network has minimal computational ability and limited knowledge of the network topology. Broadcast mechanisms are often required in such ad-hoc networks to disburse key network-related information, for example, to carry out over-the-air programming of the IoT nodes. Further, these mechanisms need to be completely distributed and must impose minimal computational burden on the nodes. Broadcast mechanisms such as flooding, although being distributed and reliable, are not efficient, since there are excessive transmissions and consequently a high energy expenditure \cite{tseng2002broadcast}. To overcome this, probabilistic forwarding of received packets may be employed (see \cite{sasson2003probabilistic},\cite{haas2006gossip}), wherein each node either forwards a previously unreceived packet to all its neighbors with probability $p$ or takes no action with probability $1-p$.
	
	 In a previous paper \cite{ncc2018:probfwding}, we studied the effect of introducing redundancy in the form of coded packets into this probabilistic forwarding protocol. We describe the setup here. Consider a large network with a particular node designated as the source. The source has $k$ message packets to send to a large fraction of nodes in the network. The $k$ message packets are first encoded into $n \ge k$ coded packets using a maximum distance separable (MDS) code (see e.g., \cite[Ch.~11]{roth2006}).  The MDS code ensures that any node that receives at least $k$ of the $n$ coded packets can retrieve the original $k$ message packets by treating the unreceived packets as erasures. The $n$ coded packets are indexed by the integers from $1$ to $n$, and the source transmits each packet to all its one-hop neighbours. All the other nodes in the network use the probabilistic forwarding mechanism: when a packet (say, packet $ \# j$) is received by a node for the first time, it either transmits it to all its one-hop neighbours with probability $p$ or does nothing with probability $1-p.$ The node ignores all subsequent receptions of packet $\# j$.
	
	Our goal is to analyze the performance of the above algorithm. In particular, we wish to find the minimum retransmission probability $p$ for which the expected fraction of nodes receiving at least $k$ out of the $n$ coded packets is close to 1, which we deem a ``near-broadcast''. 
This probability yields the minimum value for the expected total number of transmissions across all the network nodes needed for a near-broadcast. The expected total number of transmissions is a measure of the energy expenditure in the network. 
	
	Simulation results presented in \cite{ncc2018:probfwding} indicate that over a wide range of network topologies (including the important case of random geometric graphs, but not including tree-like topologies),  the expected total number of transmissions initially decreases to a minimum and then increases with $n$. Our aim is to understand this behaviour and predict, via analysis, the value of $n$ that minimizes the expected number of transmissions. While we would ultimately like to explain this behaviour on random geometric graphs, which constitute an important model for wireless ad-hoc networks \cite{vaze2015random}, we have not yet developed the tools required for the analysis there. In this paper, we present an analysis for trees and grids.
	
	The rest of the paper is organized as follows. Section-\ref{sec:formulation} provides a theoretical formulation of the problem. In Sections~\ref{sec:tree} and~\ref{sec:grid}, we consider the problem on rooted binary trees and grids, respectively, and provide bounds and estimates for the expected number of transmissions. 
The appendix contains proofs of our results.
	
	\section{Problem Formulation}\label{sec:formulation}
		The problem formulation here is essentially reproduced from \cite{ncc2018:probfwding}. 
	Consider a graph $G=(V,E)$, where $V$ is the vertex set with $N$ vertices (nodes) and $E$ is the set of edges (noiseless communication links). A source node $s \in V$ has $k$ message packets which need to be broadcast in the network. The source $s$ encodes the $k$ messages into $n$ coded packets using an MDS code (see e.g., \cite[Ch.~11]{roth2006}). Thus, on receiving any $k$ of these $n$ coded packets, a node can retrieve the original $k$ message packets. It is assumed that each packet has a header which identifies the packet index $j \in [n] := \{1,2,\ldots,n\}$.  It is also assumed that all the required encoding/decoding operations are carried out over a sufficiently large field, so that an MDS code with the necessary parameters exists. 

 The source node broadcasts all $n$ coded packets to its one-hop neighbours, after which the probabilistic forwarding protocol takes over. A node receiving a particular packet for the first time, forwards it to all its one-hop neighbours with probability $p$ and takes no action with probability $1-p$. Each packet is forwarded independently of other packets and other nodes. This probabilistic forwarding continues until there are no further transmissions in the system. The protocol indeed must terminate after finitely many transmissions since each node in the network may choose to forward a particular coded packet only the first time it is received. The node ignores all subsequent receptions of the same packet, irrespective of the decision it took at the time of first reception.

We are interested in the following scenario. Let $R_{k,n}$ be the number of nodes, including the source node, that receive at least $k$ out of the $n$ coded packets. Given a $\delta \in (0,1)$, let $p_{k,n,\delta}$ be the minimum forwarding probability $p$ for a near-broadcast, i.e.,
	\begin{equation}
	p_{k,n,\delta} \ := \ \inf\left\{p\  \bigg|\ \mathbb{E}\left[\frac{R_{k,n}}{N}\right] \geq 1-\delta \right\}.
	\label{def:pkndelta}
	\end{equation} 

The performance measure of interest, denoted by $\tau_{k,n,\delta}$, is the expected total number of transmissions across all nodes when the forwarding probability is set to $p_{k,n,\delta}$. Here, it should be clarified that whenever a node forwards (broadcasts) a packet to all its one-hop neighbours, it is counted as a single (simulcast) transmission. Our aim is to determine, for a given $k$ and $\delta$, how $\tau_{k,n,\delta}$ varies with $n$, and the value of $n$ at which it is minimized. To this end, it is necessary to first understand the behaviour of $p_{k,n,\delta}$ as a function of $n$. In this direction, we have the following simple lemma, valid for any connected graph $G = (V,E)$, proved in Section~\ref{sec:lem1proof} of the Appendix.
	
	\begin{lem}\label{lem:pkndelzero}
		For fixed values of $k$ and $\delta$, 
		\begin{enumerate}
			\item[\emph{(a)}] $p_{k,n,\delta}$ is a non-increasing function of n.
			\item[\emph{(b)}] $p_{k,n,\delta} \rightarrow 0$ as $n\rightarrow \infty.$
		\end{enumerate}
	\end{lem}
	
On the other hand, $\tau_{k,n,\delta}$ typically exhibits more complex behaviour. As demonstrated via simulations in \cite{ncc2018:probfwding}, over a wide range of graph topologies (both deterministic and random), except notably for trees (see Section~\ref{sec:tree}), $\tau_{k,n,\delta}$ initially decreases and then grows gradually as $n$ increases. This trend can be seen most clearly in a grid topology --- see Section~\ref{sec:grid}. Thus, there typically is an optimal value of $n$ that minimizes $\tau_{k,n,\delta}$. This happens due to an interplay between two opposing factors: as $n$ increases, $p_{k,n,\delta}$ decreases (Lemma~\ref{lem:pkndelzero}), which contributes towards a decrease in $\tau_{k,n,\delta}$. But this is opposed by the fact that the overall number of transmissions tends to increase when there are a larger number of packets traversing the network. To determine the value of $n$ that minimizes $\tau_{k,n,\delta}$, we need more precise estimates of $p_{k,n,\delta}$, and consequently, $\tau_{k,n,\delta}$. For specific graph topologies, we may be able to obtain such estimates using methods tailored to those topologies. We demonstrate this for two topologies in the next two sections, starting with the easiest case of a binary tree.
	
	\section{Rooted Binary Trees}\label{sec:tree}
	\begin{figure}
		\centerline{\scalebox{0.38}{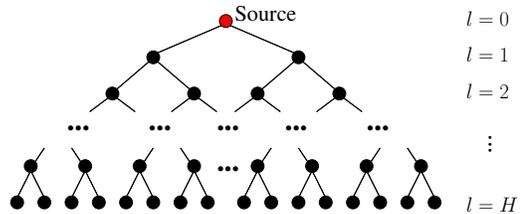}} 		
		\caption{A rooted binary tree of height $H$.}
		\label{fig:bintree}
		\vspace{-1.5em}
	\end{figure}
	Consider a rooted binary tree of height $H \ge 2$ as depicted in Fig.~\ref{fig:bintree}. The root of the tree is the source node and is at level $l=0$.
	The root node encodes the $k$ data packets into $n$ coded packets and transmits them to its children. Every other node on the tree follows the probabilistic forwarding strategy with some fixed forwarding probability $p > 0$. We will assume that the nodes at level $H$ (i.e., the leaf nodes) do not transmit, as there is nothing to be gained in allowing them to do so. 
	
	On a rooted binary tree of height $H \ge 2$, when the number of data packets, $k$, is fixed, and $\delta > 0$ is sufficiently small, simulation results presented in \cite{ncc2018:probfwding} show that $\tau_{k,n,\delta}$ increases with $n$. A large-deviations analysis aiming to explain these results was also attempted in \cite{ncc2018:probfwding}. However, this analysis is only valid in the regime where $k$ and $n$ are both large but the ratio $k/n$ is fixed. Thus, the analysis in  \cite{ncc2018:probfwding} does not in fact explain the simulation results. 
	
	In this section, we present results for the regime when $k$ is fixed and $n$ is allowed to vary. We also fix a $H \ge 2$ and a $\delta \in (0,\frac18)$. In the interest of brevity, and in order to quickly move on to the more interesting and challenging analysis for grids in the next section, we only provide the statements of the results here. Detailed derivations of these results can be found in Section~\ref{sec:bintree_proofs} of the Appendix. 
	
	It was shown in \cite{ncc2018:probfwding} that $p_{k,n,\delta}$ is the least value of $p \in [0,1]$ for which\footnote{This is a re-arrangement of Eq.~(4) in \cite{ncc2018:probfwding}.} 
	\begin{equation}\label{exprtomin}
	\frac{\sum_{l=0}^{H-1}2^{l+1}\mathbb{P}(Z_l \leq k-1)}{2^{H+1}-1} \leq \delta,
	\end{equation}
	where $Z_l \sim \text{Bin}(n,p^l)$ for $l = 0,1,\ldots,H-1$, and 
	 \begin{equation}
	\tau_{k,n,\delta} \ = \  n \, \left[\frac{(2p_{k,n,\delta})^H-1}{2p_{k,n,\delta}-1}\right].
	\label{eq:tau_tree}
	\end{equation}

An analysis starting from \eqref{exprtomin} yields the two propositions below, which provide good lower and upper bounds on $p_{k,n,\delta}$. These bounds are plotted, for $k=100$, $\delta = 0.1$ and $H = 50$, in Fig.~\ref{fig:tree_plots}(a) along with the exact values of $p_{k,n,\delta}$ obtained numerically from \eqref{exprtomin}. The corresponding plots for $\tau_{k,n,\delta}$, obtained via \eqref{eq:tau_tree}, are shown in Fig.~\ref{fig:tree_plots}(b). 
	
	\begin{prop}
		Let $k \ge 2$, $H \ge 2$, and $0 \le \delta < \frac18$ be fixed. For all $n \ge k$, we have 
		$
		p_{k,n,\delta} > {\left(\frac{k-1}{n}\right)}^{\frac{1}{H-1}}.
		$
		
		In the case of $k=1$ and $n > 1$, the lower bound can be improved to
		$
		p_{k,n,\delta} > {\left(\frac{1}{n}\right)}^{\frac{1}{H-1}}.
		$
		\label{prop:p_lobnd}
	\end{prop}

	\begin{prop}
		Let $k \ge 2$, $H \ge 2$, and $0 < \delta \le 1$ be fixed, and let $\delta' :=  \min\left\{\delta  \left(\frac{2^{H+1}-1}{2^{H+1}-2}\right), 1 \right\}$. Then, for all $n \ge 1$, we have
		$$
		p_{k,n,\delta} \le \min\left\{{\left(\frac{k-1+t}{n}\right)}^{\frac{1}{H-1}},1\right\}, 
		$$
		where $t = \sqrt{2(k-1)(-\ln\delta') + (\ln\delta')^2} - \ln\delta'$.
		In the case of $k=1$, the bound
		$$
		p_{k,n,\delta} \le \min\left\{{\left(\frac{-\ln\delta'}{n}\right)}^{\frac{1}{H-1}}, 1\right\}
		$$
		holds for all $n \ge 1$.
		\label{prop:p_upbnd}
	\end{prop}
	
 The following theorem, which summarizes the behaviour of $p_{k,n,\delta}$ on binary trees, is a direct consequence of Propositions~\ref{prop:p_lobnd} and~\ref{prop:p_upbnd}.
 
	\begin{thm}\label{thm:pkndelta}
		Let $k \ge 2$, $H \ge 2$ and $0 < \delta < \frac18$ be fixed. We then have
		$
		p_{k,n,\delta} = \Theta\left({\textstyle{\left(\frac{k}{n}\right)}^{\frac{1}{H-1}}}\right),
		$
		where the constants implicit in the $\Theta$-notation\footnote{The notation $a(n) = \Theta(b(n))$ means that there are positive constants $c_1$ and $c_2$ such that $c_1 b(n) \le a(n) \le c_2 b(n)$ for all sufficiently large $n$.} may be chosen to depend only on $H$ and $\delta$.
	\end{thm}

The plots in Figure~\ref{fig:tree_plots} corroborate the simulation results reported in \cite{ncc2018:probfwding}, thus providing a theoretical explanation for why $\tau_{k,n,\delta}$ increases with $n$. Another confirmation of this behaviour can be obtained by substituting $p_{k,n,\delta} \approx c {\bigl(\frac{k}{n}\bigr)}^{\frac{1}{H-1}}$, for any positive constant $c \equiv c(H,\delta)$ into the expression for $\tau_{k,n,\delta}$ in \eqref{eq:tau_tree}. This yields the approximation 
	\begin{equation}
	\tau_{k,n,\delta}  \ \approx  \ n \, \left[ \frac{(2c)^H \left(\frac{k}{n}\right)^{\frac{H}{H-1}} - 1}{2c\left(\frac{k}{n}\right)^{\frac{1}{H-1}}-1}\right], \label{taukrho2}
	\end{equation}
	which can be shown to be increasing in $n$.
	\begin{figure} 
		\centering
		\subfloat[Minimum retransmission probability]{%
			\includegraphics[width=0.95\linewidth]{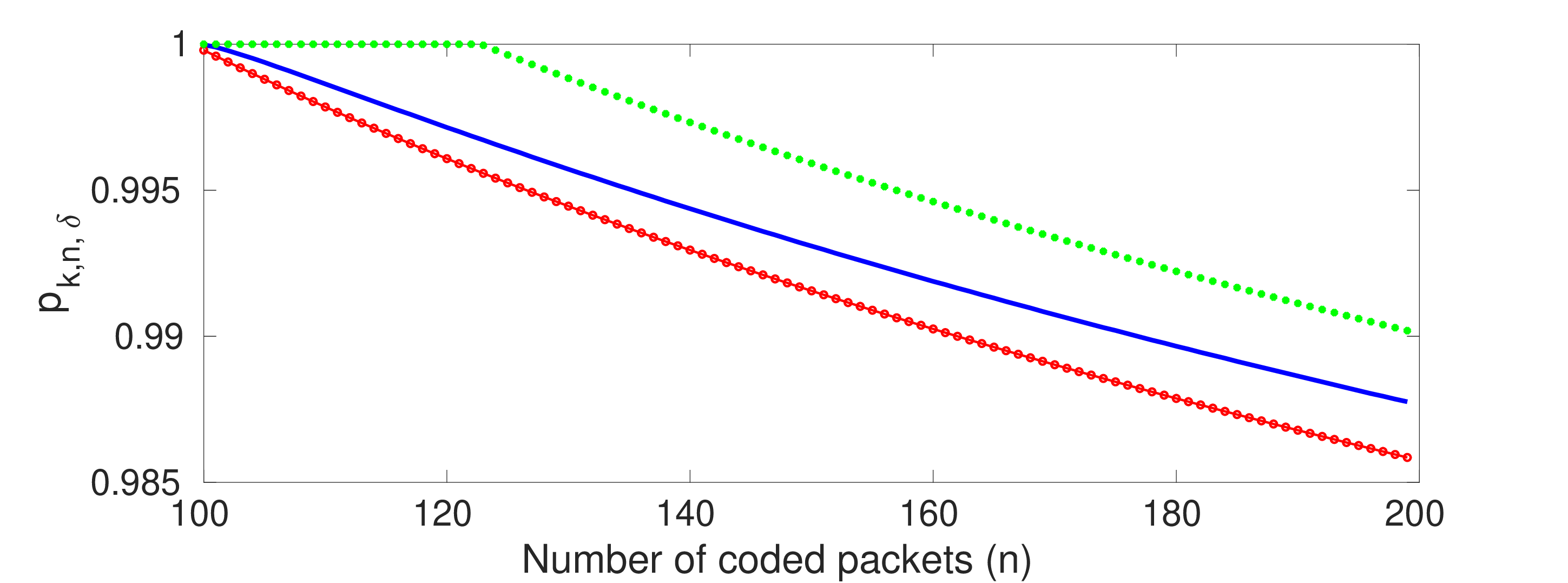}}
		\label{4a}
		\subfloat[Expected total number of transmissions]{%
			\includegraphics[width=0.95\linewidth]{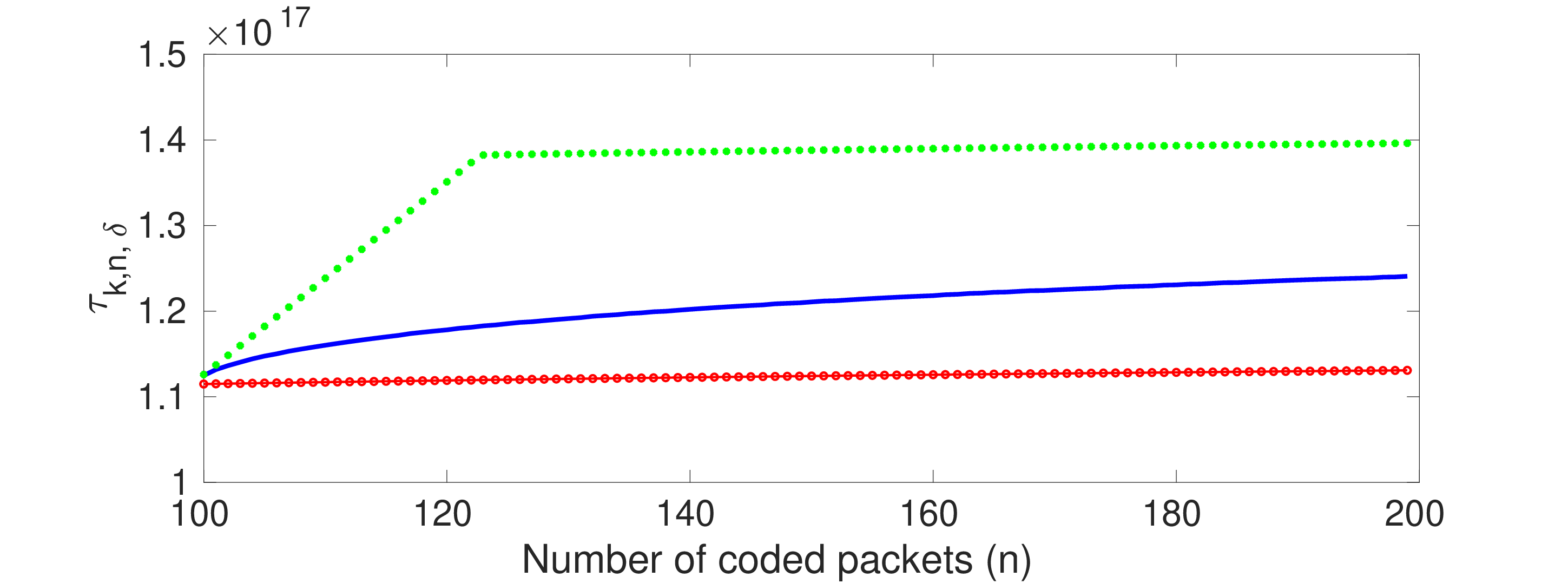}}
		\label{4b}\\
		\caption{The middle curves are plots of the true values of $p_{k,n,\delta}$ and $\tau_{k,n,\delta}$ obtained from (\ref{exprtomin}) and \eqref{eq:tau_tree}, for $k = 100$, $\delta = 0.1$ and $H=50$. The other curves are bounds obtained via Propositions \ref{prop:p_lobnd} and \ref{prop:p_upbnd}, and \eqref{eq:tau_tree}.}
		\label{fig:tree_plots} 
	\end{figure}
	
The analysis in this section extends easily to the case of rooted $d$-ary trees, for any $d \ge 2$. In summary, introducing redundancy in the form of coding into the probabilistic retransmission protocol on a rooted $d$-ary tree is not beneficial in terms of the overall energy expenditure in the network.
	
\section{Grids}\label{sec:grid}

	\begin{figure}
		\centering
		\includegraphics[width=0.4\linewidth]{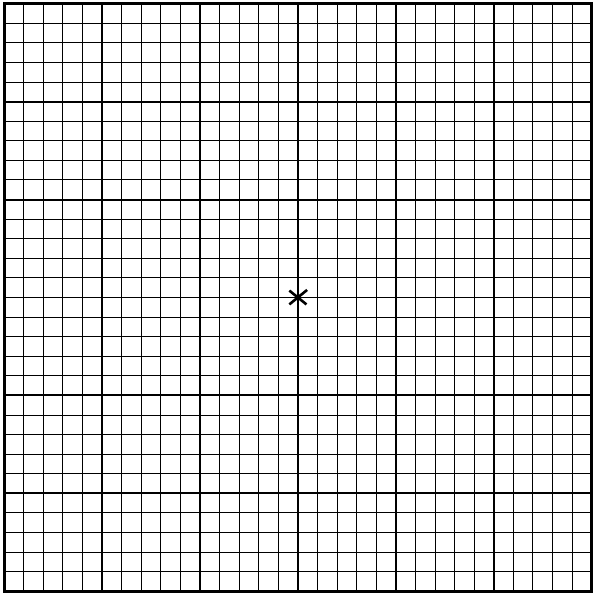}
		\caption{The source node ($\times$) is at the centre of the $31 \times 31$ grid.}
		\label{fig:grid31}
		\vspace{-1.5em}
	\end{figure}%
	
	For odd $m$, we consider the $m \times m$ grid $\Gamma_m := {[-\frac{m-1}{2},\frac{m-1}{2}]}^2 \cap \Z^2$ centred at the origin. The source node is assumed to be at the centre of the grid.  Fig.~\ref{fig:grid31} depicts this for $m=31$.  Simulation results for the probabilistic forwarding algorithm on grids of various sizes were presented in \cite{ncc2018:probfwding}. Some results from simulations on the $\G_{31}$ and $\G_{501}$ grids are shown in Figs.~\ref{fig:ergodic_pkndelta} and~\ref{fig:ergodic_trans} in Section~\ref{sec:grid_analysis}. In this section, we try to explain these observations by developing an analysis that is at least valid for large $m$. Specifically, we turn to the theory of site percolation on the integer lattice $\Z^2$ to explain the $p_{k,n,\delta}$ and $\tau_{k,n,\delta}$ curves obtained via simulations on large grids $\G_m$.
	
	
\subsection{Site percolation on $\Z^2$} \label{sec:perc}
	
	We start with a brief description of the site percolation process (see e.g.\ \cite{grimmett}) on $\Z^2$. This is an i.i.d.\ process ${(X_u)}_{u\in\Z^2}$, with $X_u \sim \text{Ber}(p)$ for each $u \in \Z^2$, where the probability $p \in [0,1]$ is a parameter of the process. A node or \emph{site} $u \in \Z^2$ is \emph{open} if $X_u = 1$, and is \emph{closed} otherwise. For $u = (u_x,u_y) \in \Z^2$, define $|u| := |u_x|+|u_y|$. Two sites $u$ and $v$ are joined by an edge, denoted by $u$---$v$, iff $|u-v| = 1$. The next few definitions are made with respect to a given realization of the process ${(X_u)}_{u\in\Z^2}$. Two sites $u$ and $v$ are connected by an \emph{open path}, denoted by $u \longleftrightarrow v$, if there is a sequence of sites $u_0 = u, u_1,u_2,\ldots,u_n = v$ such that $u_k$ is open for all $k \in \{0,1,\ldots, n\}$ and $u_{k-1}$---$u_k$ for all $k \in [n]$. The \emph{open cluster}, $C_u$, containing the site $u$ is defined as $C_u=\{v\in \mathbb{Z}^2 | u \longleftrightarrow v\}$. Thus, $C_u$ consists of all sites connected to $u$ by open paths. In particular, $C_u = \emptyset$ if $u$ is itself closed. The \emph{boundary}, $\partial C_u$, of a non-empty open cluster $C_u$ is the set of all closed sites $v \in \Z^2$ such that \mbox{$v$---$w$} for some $w \in C_u$. The set $C_u^+ := C_u \cup \partial C_u$ is called an \emph{extended cluster}. The cluster $C_u$ (resp.\ $C_u^+$) is termed an \emph{infinite open cluster (IOC)} (resp.\ \emph{infinite extended cluster (IEC)}) if it has infinite cardinality. Note that $C_u^+$ is infinite iff $C_u$ is infinite.

It is well-known that there exists a \emph{critical probability} $p_c \in (0,1)$ such that for all $p < p_c$, there is almost surely\footnote{with respect to the product measure $\otimes_{u} \nu_u$, with $\nu_u \sim \text{Ber}(p) \ \forall\, u \in \Z^2$.} no IOC, while for all $p > p_c$, there is almost surely a unique IOC. We do not know what happens at $p = p_c$, as the exact value of $p_c$ is itself not known (for site percolation on $\Z^2$). It is believed that $p_c \approx 0.59$ \cite[Ch.~1]{grimmett}. Another quantity of interest, which will play a crucial role in our analysis, is the \emph{percolation probability} $\theta(p)$, defined to be the probability that the origin $\0$ is in an IOC. In our analysis, we also consider the probability, $\theta^+(p)$, of the origin $\0$ being in an IEC. Clearly, for $p < p_c$, we have $\theta^+(p) = \theta(p) = 0$; for $p > p_c$, it is not difficult to see that $\theta^+(p) \ge \theta(p) > 0$. It is known that $\theta(p)$ is non-decreasing and infinitely differentiable in the region $p>p_c$ \cite{russo1978note}, but there is no analytical expression known for it. The following lemma expresses $\theta^+(p)$ in terms of $\theta(p)$. 

\begin{lem}
For any $p > p_c$, we have $\theta^+(p)=\frac{\theta(p)}{p}$.
\label{lem:theta}
\end{lem} 
\begin{proof}
Let $C$ and $C^+$ be the (unique) IOC and IEC, respectively. We then have
\begin{equation}
\theta(p) = \P(\0 \in C) = \P(\0 \in C^+ \text{ and } \0 \text{ is open}).
\label{eq:theta}
\end{equation}
Now, observe that the event $\{\0 \in C^+\}$ is determined purely by the states of the nodes other than the origin. Hence, this event is independent of the event that $\0$ is open. Thus, the right-hand side (RHS) of \eqref{eq:theta} equals $\theta^+(p) \cdot p$, which proves the lemma.
\end{proof}

\begin{figure}
	\centering
	\includegraphics[width=0.5\textwidth]{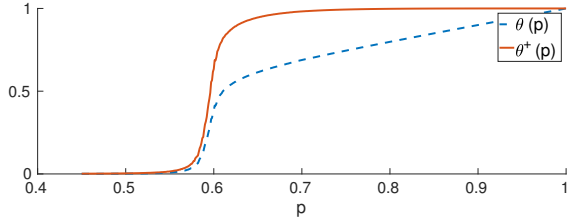}
	\caption{$\theta(p)$ and $\theta^+(p)$ vs. $p$}
	\label{fig:theta}
\end{figure}%

Fig.~\ref{fig:theta} plots $\theta(p)$ and $\theta^+(p)$ as functions of $p$, the former being obtained via simulations based on the theorem below. 

\begin{thm}\label{thm:theta}
Let $p>p_c$, and let $C$ and $C^+$, respectively, be the (almost surely) unique IOC and IEC of a site percolation process on $\mathbb{Z}^2$ with parameter $p$. Then, almost surely, we have
$$\lim_{m\rightarrow \infty}\frac{1}{m^2} |C\cap \G_m|=\theta(p) \ \text{ and } \ 
  \lim_{m\rightarrow \infty}\frac{1}{m^2}|C^+\cap \G_m|=\theta^+(p).$$
\end{thm}

The theorem is obtained as a straightforward application of an ergodic theorem for multi-dimensional i.i.d.\ random fields \cite[Proposition~8]{newman1981infinite} --- see Section~\ref{sec:ergthms} in the Appendix. Using the dominated convergence theorem (DCT), we also have
${\displaystyle \lim_{m\rightarrow \infty}} \E\left[\frac{1}{m^2}|C\cap \G_m|\right]=\theta(p)$ and
  ${\displaystyle \lim_{m\rightarrow \infty}} \E\left[\frac{1}{m^2}|C^+\cap \G_m|\right]=\theta^+(p).$
Based on this, to obtain an estimate of $\theta(p)$, the site percolation process with parameter $p$ was simulated on a $501 \times 501$ grid and the average fraction of nodes (averaged over $100$ realizations of the process) in the largest open cluster was taken to be the value of $\theta(p)$. These are the values of $\theta(p)$ plotted in Fig.~\ref{fig:theta}. We would like to emphasize that the plots in the figure should only be trusted for $p > p_c$, as Theorem~\ref{thm:theta} is only valid in that range. However, as the exact value of $p_c$ is unknown, simulation results are reported for the range of $p$ values shown in the plot.

\subsection{Relating site percolation to probabilistic forwarding} \label{sec:perc_to_fwding}
Site percolation on $\Z^2$ is a faithful model for probabilistic forwarding of a single packet on the infinite lattice $\Z^2$. The origin $\0$ is the source of the packet. The open cluster, $C_{\0}$, containing the origin $\0$ corresponds to the set of nodes that transmit (forward) the packet, and the extended cluster $C^+_{\0}$ corresponds to the set of nodes that receive the packet. The only caveat is that, since the source is assumed to always transmit the packet, we must consider only those realizations of the site percolation process in which the origin $\0$ is open. In other words, we must consider the site percolation process conditioned on the event that the origin is open. By extension, the probabilistic forwarding of $n$ coded packets on the lattice $\Z^2$ corresponds to $n$ independent site percolation processes on $\Z^2$, conditioned on the event that the origin is open in all $n$ percolations. 

\subsection{Analysis of probabilistic forwarding on a large (finite) grid} \label{sec:grid_analysis}
Our analysis of probabilistic forwarding on the finite grid $\G_m$ is based on the approximation described next. For the purposes of this discussion, we fix a forwarding probability $p$. Let $\cR_{k,n}^{\infty}$ denote the set of all nodes that receive at least $k$ of the $n$ coded packets in the probabilistic forwarding protocol on $\Z^2$. We use $|\cR_{k,n}^{\infty} \cap \G_m|$ as a proxy for $R_{k,n}^m$, which, as in Section~\ref{sec:formulation}, is defined to be the number of nodes receiving at least $k$ out of $n$ packets in the probabilistic forwarding protocol on $\G_m$. In general, it is only true that $R_{k,n}^m$ is stochastically dominated\footnote{A random variable $X$ is stochastically dominated by a random variable $Y$ if $\P(X \ge x) \le \P(Y \ge x)$ for all $x \in \R$. For non-negative random variables, this implies that $\E[X] \le \E[Y]$.} by $|\cR_{k,n}^{\infty} \cap \G_m|$, since a node in $\cR_{k,n}^\infty \cap \G_m$ could receive packets from the origin through paths in $\Z^2$ that do not lie entirely within $\G_m$. Nonetheless, we proceed under the assumption that $\E[R_{k,n}^m] \approx \E[|\cR_{k,n}^{\infty} \cap \G_m|]$ for large $m$. This is vindicated by the fact that our analysis based on this assumption matches the simulation results reasonably well --- see Figs.~\ref{fig:ergodic_pkndelta} and~\ref{fig:ergodic_trans}.
	
Recall also that we want values of the forwarding probability $p$ for which $\E[\frac{1}{m^2} R_{k,n}^m]$ is at least $1-\delta$, for some (small) $\delta > 0$. Hence, we need $\E[\frac{1}{m^2} |\cR_{k,n}^{\infty} \cap \G_m| ]\ge 1-\delta$. If we would like this to hold for all sufficiently large $m$, then $p$ must be such that $\cR_{k,n}^\infty$ has infinite cardinality. This implies, due to the correspondence between probabilistic forwarding and site percolation on $\Z^2$, that $p$ must be such that there exist infinite (open/extended) clusters in the site percolation process. Thus, we must operate in the super-critical region $p > p_c$. It can also be seen from the simulation results in Figs.~\ref{fig:ergodic_pkndelta} and \ref{fig:ergodic_trans} that $\tau_{k,n,\delta}$ is minimized when $p_{k,n,\delta}$ is in the super-critical region. We use these arguments as justification for considering only the $p > p_c$ case in our analysis.

The following theorem is the main result of this section. The proof, given in Section~\ref{sec:grid_proofs} of the Appendix, is obtained by carefully relating $\cR_{k,n}^\infty$ to the set, $C_{k,n}^+$, of all sites in $\Z^2$ that belong to the IEC containing $\0$ in at least $k$ out of $n$ independent percolations, conditioned on $\0$ being open in all $n$ percolations.

\begin{thm}
For $p>p_c$, we have
	\begin{align*}
	\lim_{m\rightarrow\infty} & \E\left[\frac{1}{m^2} \,|\cR^\infty_{k,n} \cap \G_m|\right] \\
	& \ \ \ \ \ \ = \ \sum_{t=k}^{n} \sum_{j=k}^{t}\binom{n}{t}\binom{t}{j}(\theta^+(p))^{t+j}(1-\theta^+(p))^{n-j} \; .
	\end{align*}
	\label{thm:ERkn}
\end{thm}

From the discussion prior to the theorem, the left-hand side (LHS) of the equality stated in the theorem is our proxy for ${\displaystyle \lim_{m \to \infty}} \E[\frac{1}{m^2} \, R_{k,n}^m]$. Thus, for large grids $\G_m$, we take $p_{k,n,\delta}$ to be the least value of $p$ for which 
\begin{equation}
\sum_{t=k}^{n} \sum_{j=k}^{t}\binom{n}{t}\binom{t}{j}(\theta^+(p))^{t+j}(1-\theta^+(p))^{n-j} \ \ge \ 1- \delta \, .
\label{eq:pkndelta_grid}
\end{equation}
This can be evaluated numerically using the values of $\theta^+(p)$ plotted in Fig.~\ref{fig:theta}. The results thus obtained are shown in Fig.~\ref{fig:ergodic_pkndelta}. It is clear that these results match very well with those obtained from simulations on a $501 \times 501$ grid. 

	\begin{figure}[t]
		\centering
		\includegraphics[width=0.5\textwidth]{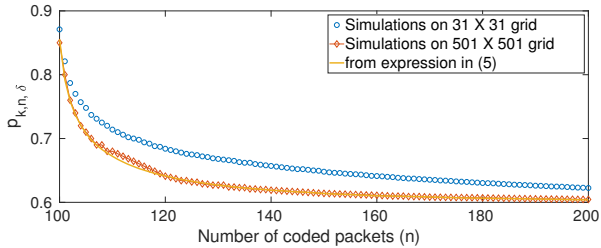}
		\caption{Comparison of the minimum forwarding probability obtained via simulations on a $31\times 31$ grid and a $501\times 501$ grid, with the results obtained numerically from \eqref{eq:pkndelta_grid}, for $k=100$ data packets and $\delta=0.1$.}
		\label{fig:ergodic_pkndelta}
		\vspace*{-1em}
	\end{figure}%

We next look into estimating the expected total number of transmissions at a given forwarding probability $p$. Consider the transmission of a single packet on the infinite lattice $\mathbb{Z}^2$. The set of nodes transmitting this packet is simply the open cluster $C_{\0}$ in the percolation framework. Thus, arguing as for packet receptions above, the expected number of transmissions for probabilistic forwarding on a large (but finite) grid $\G_m$ is well-approximated by $\E\bigl[|C_{\0} \cap \G_m| \ \big| \ \0 \text{ is open}\bigr]$. In Section~\ref{sec:grid_proofs} of the Appendix, we prove the following result.

\begin{prop} For site percolation with $p > p_c$, we have
$$
 \lim_{m \to \infty} \frac{1}{m^2}\E\bigl[|C_{\0} \cap \G_m| \ \big| \ \0 \text{ is open}\bigr] \ = \ \frac{{\theta(p)}^2}{p}.
 $$
 \label{prop:C0}
 \end{prop}
 
 Thus, in probabilistic forwarding of a single packet on a large grid $\G_m$, the expected number of transmissions, normalized by the grid size $m^2$, is approximately $\frac{{\theta(p)}^2}{p}$. Hence, when we have $n$ coded packets, by linearity of expectation, 
 the expected total number of transmissions, again normalized by the grid size $m^2$, is approximately $n \, \frac{{\theta(p)}^2}{p}$. In particular, setting $p = p_{k,n,\delta}$, we obtain
	\begin{equation}\label{eq:taukndelta}
	\frac{1}{m^2} \, \tau_{k,n,\delta}\approx n \frac{{\theta(p_{k,n,\delta})}^2}{p_{k,n,\delta}},
	\end{equation}
provided that $p_{k,n,\delta} > p_c$.

	\begin{figure}
		\centering
		\includegraphics[width=0.5\textwidth]{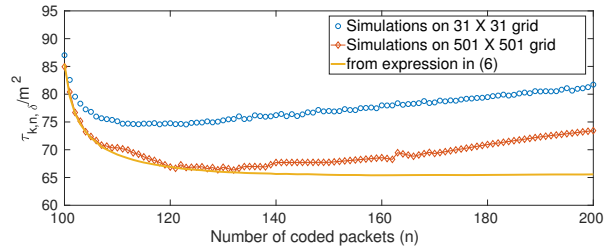}
		\caption{Comparison of the expected total number of transmissions normalized by the grid size $m^2$, obtained via simulations on $\G_{31}$ and $\G_{501}$, with the expression from \eqref{eq:taukndelta}, for $k=100$ data packets and $\delta=0.1$.}
		\label{fig:ergodic_trans}
		\vspace*{-1em}
	\end{figure}%
	
Fig.~\ref{fig:ergodic_trans} compares, for $k=100$ data packets and $\delta = 0.1$, the values of $\frac{1}{m^2} \tau_{k,n,\delta}$ obtained using \eqref{eq:taukndelta} with those obtained via simulations on the $\G_{31}$ and $\G_{501}$ grids. The curve based on \eqref{eq:taukndelta} initially tracks the $\G_{501} $ curve well, but trails off after $n = 130$. Given that the corresponding $p_{k,n,\delta}$ curves are well-matched (Fig.~\ref{fig:ergodic_pkndelta}), this is perhaps attributable to the fact that the $\theta(p)$ values from Fig.~\ref{fig:theta} are not very reliable. As even a small change in $\theta(p)$ would significantly affect $n \frac{{\theta(p)}^2}{p}$, better estimates of $\theta(p)$ may correct the discrepancy observed. Nonetheless,  our analysis provides theoretical validation, at least for large grids, for the observed initial decrease in $\tau_{k,n,\delta}$ as a function of $n$, thus indicating a benefit to introducing some coding into the probabilistic forwarding mechanism on grids.

	
		\section*{Acknowledgements} The research presented in this paper was supported in part by a Cisco PhD Fellowship awarded to the first author, and in part by the DRDO-IISc ``Frontiers'' research programme.
	
	\bibliographystyle{IEEEtran}
	\bibliography{references}

	
	
	\appendix
	
\subsection{Proof of Lemma \ref{lem:pkndelzero}} \label{sec:lem1proof}
		(a)\ For any $n > 0$, the random variables $R_{k,n}$ and $R_{k,n-1}$ can be coupled as follows: If the $k$ data packets are encoded into $n$ coded packets, then $R_{k,n-1}$ (resp.\ $R_{k,n}$) is realized as the number of nodes, including the source node, that receive at least $k$ of the \emph{first} $n-1$ (resp.\ at least $k$ of the $n$) coded packets. It is then clear that $\E[\frac{1}{N} R_{k,n}] \ge \E[\frac{1}{N} R_{k,n-1}]$, and hence, by \eqref{def:pkndelta}, we have $p_{k,n,\delta} \le p_{k,n-1,\delta}$.
		
		(b)\  From the $n$ coded packets, create $\lfloor\frac{n}{k}\rfloor$ non-overlapping (i.e., disjoint)  groups of $k$ packets each. For $i=1,2,\cdots,\floor{\frac{n}{k}}$, let $A_i$ be the event that the $i$th group of $k$ coded packets is received by at least $(1-\delta/2)N$ nodes. The events $A_i$ are mutually independent and have the same probability of occurrence. For any $p>0$, we have $\P(A_i)$ being strictly positive (but perhaps small). Hence, $$\P(\text{at least one } A_i \text{ occurs}) = 1-\bigl(1-\P(A_1)\bigr)^{\lfloor\frac{n}{k}\rfloor} \ge 1-\frac{\delta}{2}$$ for all sufficiently large $n$, so that $\P\left(\frac{R_{k,n}}{N}\ge 1-\delta/2\right)\ge1-\delta/2.$ This further implies that $\frac{\E\left[R_{k,n}\right]}{N}\ge (1-\delta/2)(1-\delta/2) \ge 1-\delta$. Thus, for any $p>0$, we have $p_{k,n,\delta}\le p$ for all sufficiently large $n$.	\endIEEEproof

\subsection{Derivations of results for binary trees} \label{sec:bintree_proofs}
	\begin{IEEEproof}[Proof of Proposition \ref{prop:p_lobnd}]
		Suppose that $p$ is such that $np^{H-1} \le k-1$. Then, $Z_{H-1}$ has mean at most $k-1$. As a result, the median of $Z_{H-1}$ is also at most $k-1$ \cite[Corollary~3.1]{jogdeo1968monotone}. In other words, $P(Z_{H-1} \le k-1) \ge \frac12$. Consequently, $\sum_{l=0}^{H-1}2^{l+1}\mathbb{P}(Z_l \leq k-1) \ge 2^H \mathbb{P}(Z_{H-1} \le k-1) \ge 2^{H-1}$, so that the LHS of \eqref{exprtomin} is at least $\frac{2^{H-1}}{2^{H+1}-1} \ge \frac{2^{H-1}}{2^{H+1}} = 0.25 > \delta$. Hence, for \eqref{exprtomin} to hold, we must have $np^{H-1} > k-1$, from which the lower bound on $p_{k,n,\delta}$ follows.
		
		In the case of $k=1$, suppose that $p \le {\left(\frac{1}{n}\right)}^{H-1}$. Then, $\mathbb{P}(Z_{H-1} = 0) = (1-p^{H-1})^n \ge (1-\frac{1}{n})^n \ge (1-\frac{1}{2})^2 = 0.25$, for all $n \ge 2$. Hence, $\sum_{l=0}^{H-1}2^{l+1}\mathbb{P}(Z_l \leq k-1) \ge 2^H \mathbb{P}(Z_{H-1} = 0) \ge 2^{H-2}$. As a result, the LHS of \eqref{exprtomin} is at least $\frac{2^{H-2}}{2^{H+1}} = 0.125 > \delta$. Thus, again, for \eqref{exprtomin} to hold, we need $p > {\left(\frac{1}{n}\right)}^{H-1}$.	
		\end{IEEEproof}
	
	\bigskip
	
	\begin{IEEEproof}[Proof of Proposition \ref{prop:p_upbnd}]
		Note first that for all $l \le H-1$, we have\footnote{This is easily shown by a standard coupling argument --- see e.g., \cite[Lemma~IV.1]{ncc2018:probfwding}.} $\mathbb{P}(Z_l \le k-1) \le \mathbb{P}(Z_{H-1} \le k-1)$. Hence, $\sum_{l=0}^{H-1}2^{l+1}\mathbb{P}(Z_l \leq k-1) \le \bigl(\sum_{l=0}^{H-1} 2^{l+1}\bigr) \mathbb{P}(Z_{H-1} \le k-1) = (2^{H+1}-2) \mathbb{P}(Z_{H-1} \le k-1)$. Thus, to show that \eqref{exprtomin} holds, it suffices to prove that $\mathbb{P}(Z_{H-1} \le k-1) \le \delta \, \left(\frac{2^{H+1}-1}{2^{H+1}-2}\right)$. It is, therefore, enough to show that $\mathbb{P}(Z_{H-1} \le k-1) \le \delta'$.
		
		Consider $k = 1$ first. Take $p = \min\left\{1,{\left(\frac{C'}{n}\right)}^{\frac{1}{H-1}}\right\}$, where $C' = -\ln\delta'$. Then, $\mathbb{P}(Z_{H-1} \le k-1) = \mathbb{P}(Z_{H-1} = 0) = (1-p^{H-1})^n$, which, by choice of $p$, is either equal to $0$ (if $n \le C'$) or $(1-C'/n)^n$ (if $C' > n$). In either case, $\mathbb{P}(Z_{H-1} = 0) $ is less than $e^{-C'} = \delta'$, as needed.
		
		Consider $k \ge 2$ now. Take $p = \min\left\{1,{\left(\frac{k-1+t}{n}\right)}^{\frac{1}{H-1}}\right\}$, where $t$ is as in the statement of the proposition. For $n \ge k-1+t$, we have $Z_{H-1} \sim \text{Bin}(n,\frac{k-1+t}{n})$, so that \begin{align*}
		\mathbb{P}(Z_{H-1} \le k-1) & = \mathbb{P}\bigl(Z_{H-1} \le n({\textstyle \frac{k-1+t}{n} - \frac{t}{n}})\bigr) \notag \\
		& \le  e^{-n \, D(\frac{k-1}{n} \parallel \frac{k-1+t}{n})}
		\end{align*}
		via the Chernoff bound. Here, $D(\cdot \parallel \cdot)$ denotes the Kullback-Leibler divergence, defined as $D(x \parallel y) = x \ln \frac{x}{y} + (1-x) \ln \frac{1-x}{1-y}$. Using the bound $D(x \parallel y) \ge \frac{(x-y)^2}{2y}$, valid for $x \le y$ \cite{Okamoto1959}, we further have
		$$
		\mathbb{P}(Z_{H-1} \le k-1) \le e^{-n\left[\frac{(t/n)^2}{2(k-1+t)/n}\right]} = e^{-\frac{t^2}{2(k-1+t)}}.
		$$
		Thus, to conclude that $\mathbb{P}(Z_{H-1} \le k-1) \le \delta'$, as required, it suffices to show that $\frac{t^2}{2(k-1+t)} \ge -\ln \delta'$. This can be re-written as $t^2 + 2t \ln\delta' + 2(k-1)\ln\delta' \ge 0$, or equivalently, $(t+\ln\delta')^2 + 2(k-1)\ln\delta' - (\ln\delta')^2 \ge 0$, which is evidently satisfied by our choice of $t$. 	
		\end{IEEEproof}

\subsection{Ergodic theorems} \label{sec:ergthms}
	Let $\mathsf{A}$ be a finite alphabet, and $\nu$ a probability measure on it. Consider the probability space $(\Omega,\mathcal{F},\P)$, where $\Omega = \mathsf{A}^{\Z^2}$, $\mathcal{F}$ is the $\sigma$-algebra of cylinder sets, and $\P$ is the product measure $\otimes_{u} \nu_u$ with $\nu_u = \nu$ for all $u \in Z^2$. For $z \in \Z^2$, define the shift operator $T_z: \Omega \to \Omega$ that maps $\omega = {(\omega_u)}_{u \in \Z^2}$ to $T_z\omega$ such that $(T_z\omega)_u = \omega_{u-z}$ for all $u \in \Z^2$. Correspondingly, for a random variable $X$ defined on this probability space, set $T_zX := X \circ T_{-z}$, i.e., $(T_zX)(\omega) = X(T_{-z}\omega)$ for all $\omega \in \Omega$.
	
	The following theorem is a special case of Tempelman's pointwise ergodic theorem (see e.g., \cite[Ch.~6]{krengel1985}). For $\mathsf{A} = \{0,1\}$, this was stated as Proposition~8 in \cite{newman1981infinite}. 
	
\begin{thm}
For any random variable $X$ on $(\Omega,\mathcal{F},\P)$ with finite mean, we have 
$$
\lim_{m \to \infty} \frac{1}{m^2} \sum_{z \in \G_m} T_zX \, = \, \E[X]  \ \ \ \ \ \ \P\text{-a.s.},
$$
where $\G_m := [-\frac{m-1}{2},\frac{m-1}{2}]^2 \cap \Z^2$ is the $m \times m$ grid ($m$ odd).
\label{thm:tempelman}	
\end{thm}
	
	The theorem applies to the case of site percolation, in which $\nu$ above is the Bernoulli($p$) measure on $\mathsf{A} = \{0,1\}$. Applying the theorem with $X = {\mathds{1}}_{\{\0 \in C\}}$, the indicator function of $\0$ being in the (unique when $p > p_c$) IOC $C$, and again with $X = {\mathds{1}}_{\{\0 \in C^+\}}$, we obtain Theorem~\ref{thm:theta}.
	
	Next, with $\mathsf{A} = \{0,1\}^n$ and $\nu$ the product of $n$ independent Bernoulli($p$) measures, we are in the setting of $n$ independent site percolations. In this case, taking $X$ to be the indicator function of $\0$ being in the IEC in at least $k$ of the $n$ independent percolations, and applying Theorem~\ref{thm:tempelman}, we obtain Theorem~\ref{thm:prodtheta} below.

\subsection{Derivations of results for grids} \label{sec:grid_proofs}

Consider $n$ independent site percolation processes on $\Z^2$, with parameter $p > p_c$. Let O denote the event that the origin is open in all $n$ percolations. We will use $\P^\mathrm{o}$ and $\E^\mathrm{o}$, respectively, to denote the probability measure and expectation operator conditioned on the event O, and $\P$ and $\E$ for the unconditional versions of these.

Since $p > p_c$, each percolation has a unique IOC and IEC, almost surely with respect to $\P$ ($\P$-a.s.). Let $C_{k,n}^+$ be the set of sites that are in the IEC in at least $k$ out of the $n$ percolations. We then have the following theorem.

\begin{thm} We have
$$\lim_{m \rightarrow \infty}\frac{1}{m^2} |C_{k,n}^+ \cap \G_m| = \theta^+_{k,n}(p) \ \ \ \ \ \ \P\text{-a.s.}$$
where 
$$
\theta^+_{k,n}(p) =  \sum_{j = k}^n \binom{n}{j} (\theta^+(p))^j (1-\theta^+(p))^{n-j}
$$
is the probability that the origin belongs to the IEC in at least $k$ out of the $n$ percolations.
\label{thm:prodtheta}
\end{thm} 
\begin{IEEEproof}
As discussed at the end of the last subsection, the result is a direct consequence of Theorem~\ref{thm:tempelman}. By the fact that the $n$ percolations are mutually independent, we have $\theta^+_{k,n}(p)=\P (Y_n \ge k)$, where $Y_n \sim \text{Bin}(n,\theta^+(p))$.
\end{IEEEproof}
	
From the theorem, we derive a useful fact that plays a key role in our proof of Theorem~\ref{thm:ERkn}. Since the event, say $A_n$, that the origin is in the IOC in all $n$ percolations has positive probability ($\theta(p)^n > 0$ for $p > p_c$), the theorem statement also holds almost surely when conditioned on $A_n$. Hence, by the DCT, we also have
\begin{equation}
\lim_{m \rightarrow \infty} \E\left[\frac{1}{m^2} |C_{k,n}^+ \cap \G_m| \ \bigg| \ A_n \right] = \theta^+_{k,n}(p) \, .
\label{eq:EgivenAn}
\end{equation}

Now, for $T \subseteq [n]$, define $A^+_T$ to be the event that the origin is in the IEC in exactly the percolations indexed by $T$. The following proposition relates the probability of the event $A^+_T$, conditioned on the event O, to $\theta^+(p)$.
\begin{prop} For any $T \subseteq [n]$ with $|T| = t$, we have
$$\P^{\mathrm{o}}(A^+_T)\ =\ (\theta^+(p))^{t}(1-\theta^+(p))^{n-t}$$.
\label{prop:formulap}
\end{prop}
\begin{IEEEproof}
By definition, $\P^{\mathrm{o}}(A_T^+) = \P(A_T^+ \mid \mathrm{O})$. Note that, in a given percolation, conditioned on $\0$ being open, the event $\{\0 \text{ is in the IEC}\}$ is the same as the event $\{\0 \text{ is in the IOC}\}$. Consequently, conditioned on O, the event $A_T^+$ is the same as the event, $A_T$, that the origin is in the IOC in exactly the percolations indexed by $T$. Hence, 
$$\P^{\mathrm{o}}(A_T^+) = \P(A_T \mid \mathrm{O}) = \frac{\P(A_T \cap \mathrm{O})}{\P(\mathrm{O})}.$$
The denominator equals $p^n$. The numerator is the event that the origin is in the IOC in exactly the percolations indexed by $T$, and is open but in a finite cluster in the remaining $n-|T|$ percolations. In a given percolation, the probability that the origin is open but in a finite cluster is $p-\theta(p)$. Thus, we have $\P(A_T \cap \mathrm{O})= (\theta(p))^{|T|}(p-\theta(p))^{n-|T|}$. The result now follows from the fact (Lemma~\ref{lem:theta}) that $\theta^+(p) = \frac{\theta(p)}{p}$.
\end{IEEEproof}

\medskip

We are now in a position to prove Theorem~\ref{thm:ERkn}, which we reproduce here for convenience.
\begin{thm}[Theorem~\ref{thm:ERkn}]
For $p>p_c$, we have
	\begin{align*}
	\lim_{m\rightarrow\infty} & \E\left[\frac{1}{m^2} \,|\cR^\infty_{k,n} \cap \G_m|\right] \\
	& \ \ \ \ \ \ = \ \sum_{t=k}^{n} \sum_{j=k}^{t}\binom{n}{t}\binom{t}{j}(\theta^+(p))^{t+j}(1-\theta^+(p))^{n-j} \; .
	\end{align*}
\end{thm}
\begin{IEEEproof}
In the framework of $n$ independent site percolations, $\cR_{k,n}^{\infty}$ is the set of sites in $\Z^2$ that are in the extended cluster containing the origin in at least $k$ of the $n$ percolations (conditioned on the origin being open). Thus, in terms of the percolation probability space, the expectation on the left-hand side (LHS) of the theorem statement is in fact the conditional expectation $\E^{\mathrm{o}}$. We then write
\begin{align}
\E^{\mathrm{o}} & \left[|\cR^\infty_{k,n} \cap \G_m|\right] \notag \\
 \ \ \ \ \ \ \ \ \ & =  \ \sum_{t = 0}^n \sum_{T \subseteq [n]: \atop |T|=t}  \E^{\mathrm{o}}  \left[|\cR^\infty_{k,n} \cap \G_m| \ \big| \ A_T^+\right] \, \P^{\mathrm{o}}(A_T^+) \,.
 \label{eq:Eo_sum}
\end{align}

Consider any summand with $|T| = t < k$. Given $A_T^+$, the origin is in the IEC in no more than $k-1$ of the percolations; hence, each site in $\cR_{k,n}^{\infty}$ must belong to the finite cluster, denoted by $C_{\0}[j]$, in the $j$th percolation for some $j \notin T$. As a result, given $A_T^+$, $\cR_{k,n}^\infty$ is contained in the union $\cup_{j \notin T} C_{\0}[j]$, which is finite $\P^{\mathrm{o}}$-a.s, so that ${\displaystyle \lim_{m \to \infty}} \frac{1}{m^2} |\cR^\infty_{k,n} \cap \G_m| = 0$ \ $\P^{\mathrm{o}}$-a.s.. Consequently, by the DCT, we have for any $T \subseteq[n]$ with $|T| < k$,
\begin{equation}
{\displaystyle \lim_{m \to \infty}} \E^{\mathrm{o}}  \left[\frac{1}{m^2}|\cR^\infty_{k,n} \cap \G_m| \ \big| \ A_T^+\right] = 0.
\label{eq:t<k}
\end{equation}

Next, consider any summand in \eqref{eq:Eo_sum} with $|T| = t \ge k$. The sites in $\cR_{k,n}^{\infty}$ can be exactly one of two types: those that belong to the extended cluster $C_{\0}^+$ in at least $k$ of the percolations indexed by $T$; and those that do not. Let $\cR_{k,T}^\infty$ be the subset of $\cR_{k,n}^{\infty}$ consisting of sites of the first type, and let $\mathcal{Q} = \cR_{k,n}^{\infty} \setminus \cR_{k,T}^{\infty}$. Thus,  
\begin{align}
 \E^{\mathrm{o}} & \left[|\cR^\infty_{k,n} \cap \G_m| \ \big| \ A_T^+\right]  \notag \\ 
  & \!\!\!\!\! = \ \E^{\mathrm{o}}\left[|\cR^\infty_{k,T} \cap \G_m| \ \big| \ A_T^+\right] + \E^{\mathrm{o}}\left[|\cQ \cap \G_m| \ \big| \ A_T^+\right].
 \label{eq:EoRkninfty}
\end{align}

Note that any site in $\cQ$ must belong to $C_{\0}^+$ in at least one percolation outside of $T$. In particular, given $A_T^+$, $\cQ$ is $\P^{\mathrm{o}}$-a.s.\ finite. Thus, arguing as in the $|T| < k$ case, we have 
\begin{equation}
\lim_{m \to \infty} \E^{\mathrm{o}}  \left[\frac{1}{m^2}| \cQ \cap \G_m| \ \big| \ A_T^+\right] = 0.
\label{eq:Q}
\end{equation}

Finally, note that 
\begin{align*}
\E^{\mathrm{o}}\left[|\cR^\infty_{k,T} \cap \G_m| \ \big| \ A_T^+\right] 
  & = \E\left[|\cR^\infty_{k,T} \cap \G_m| \ \big| \ A_T^+ \cap \mathrm{O}\right] \\
  & \stackrel{(a)}{=} \E\left[|C_{k,T}^+ \cap \G_m| \ \big| \ A_T^+ \cap \mathrm{O} \right] \\
  & \stackrel{(b)}{=} \E\left[|C_{k,T}^+ \cap \G_m| \ \big| \ A_T\right],
\end{align*}
where $A_T$ is the event that $\0$ is in the IOC in exactly the percolations indexed by $T$, and $C_{k,T}^+$ is the set of sites of $\Z^2$ that belong to the IEC in at least $k$ of the percolations indexed by $T$. The equality labeled~(a) above is due to the fact that, conditioned on $A_T^+ \cap O$, $\cR_{k,T}^\infty = C_{k,T}^+$. The equality labeled~(b) is because $A_T^+ \cap \mathrm{O} = A_T \cap \mathrm{O}$, and moreover, the event that $\0$ is open in the percolations outside $T$ is independent of the percolations indexed by $T$. 

Thus, restricting our attention to \emph{only} the percolations indexed by $T$, we can apply \eqref{eq:EgivenAn} with $n = t$ to obtain ${\displaystyle \lim_{m \to \infty}} \E\left[\frac{1}{m^2}|C_{k,T}^+ \cap \G_m| \ \big| \ A_T\right] = \theta_{k,t}^+(p)$. Hence,
\begin{equation}
 \lim_{m \to \infty} \E^{\mathrm{o}}  \left[\frac{1}{m^2}| \cR^\infty_{k,T}  \cap \G_m| \ \big| \ A_T^+\right] =  \theta_{k,t}^+(p).
\label{eq:Eo:last}
\end{equation}

Upon multiplying \eqref{eq:Eo_sum} by $\frac{1}{m^2}$, and letting $m \to \infty$, we obtain via \eqref{eq:t<k}--\eqref{eq:Eo:last}:
\begin{equation*}
\lim_{m \to \infty} \E^{\mathrm{o}} \left[\frac{1}{m^2}|\cR^\infty_{k,n} \cap \G_m|\right] = \sum_{t=k}^n \sum_{T \subseteq [n]: \atop |T| = t} \theta_{k,t}^+(p) \, \P^{\mathrm{o}}(A_T^+).
\end{equation*}

Applying Proposition~\ref{prop:formulap} completes the proof.
\end{IEEEproof}
	
\bigskip
	
Finally, we prove Proposition~\ref{prop:C0}, which we re-state here for ease of reference.

 \begin{prop}[Proposition~\ref{prop:C0}]
For site percolation with $p > p_c$, we have
$$
 \lim_{m \to \infty} \frac{1}{m^2}\E\bigl[|C_{\0} \cap \G_m| \ \big| \ \0 \text{ is open}\bigr] \ = \ \frac{{\theta(p)}^2}{p}.
 $$
 \end{prop}
 
\begin{IEEEproof}
 We use $\P^{\0}$ and $\E^{\0}$, respectively, to denote the probability measure and expectation operator conditioned on the event that the origin $\0$ is open. Let $C$ be the (unique) IOC, and $A$ the event $\{\0 \in C\}$.  Then,
 \begin{align*}
 \lim_{m \to \infty} \E^{\0} & \left[\frac{1}{m^2} |C_{\0} \cap \G_m| \right] \\
 & = \lim_{m \to \infty} \E \left[\frac{1}{m^2} |C_{\0} \cap \G_m| \ \big| \ A \right] \P^{\0}(A) \\
  & \ \ \ \ \ \ \ \ \ \ +  \lim_{m \to \infty} \E^{\0} \left[\frac{1}{m^2} |C_{\0} \cap \G_m| \ \big| \ A^c \right] \P^{\0}(A^c) 
\end{align*}

Now, given $A^c$ (i.e., $\0 \notin C$), $C_{\0}$ is $\P^{\0}$-a.s.\ finite, and so by the usual DCT argument, 
${\displaystyle \lim_{m \to \infty}} \E^{\0} \left[\frac{1}{m^2} |C_{\0} \cap \G_m| \ \big| \ A^c \right] = 0$. On the other hand, given $A$, we have $C_{\0} = C$. From Theorem~\ref{thm:theta}, we know that ${\displaystyle \lim_{m \to \infty}} \frac{1}{m^2} |C \cap \G_m| = \theta(p)$ $\P$-a.s.. Moreover, this statement holds even when the probability measure $\P$ is conditioned on $A$, since $\P(A) = \theta(p) > 0$ for $p > p_c$. So, again by the DCT, ${\displaystyle \lim_{m \to \infty}} \E[\frac{1}{m^2} |C \cap \G_m| \mid A] = \theta(p)$. We have thus shown that
$$
\lim_{m \to \infty} \E^{\0} \left[\frac{1}{m^2} |C_{\0} \cap \G_m| \right] = \theta(p) \, \P^{\0}(A).
$$

The proof is completed by observing that $\P^{\0}(A) = \frac{\P(A)}{\P({\0} \text{ is open})} = \frac{\theta(p)}{p}$.
\end{IEEEproof}

\end{document}